\documentclass[11pt]{article}%
\pdfoutput=1
\usepackage{amsfonts}
\usepackage{amsmath}
\usepackage{fullpage}
\usepackage{amssymb}
\usepackage{graphicx}
\usepackage[colorlinks=true,linkcolor=blue,citecolor=red,plainpages=false,pdfpagelabels]%
{hyperref}
\usepackage{mathabx}
\usepackage{comment}
\usepackage[amsmath,amsthm,thmmarks]{ntheorem}
\usepackage[]{color}
\providecommand{\U}[1]{\protect\rule{.1in}{.1in}}
\newtheorem{theorem}{Theorem}

\newtheorem{corollary}[theorem]{Corollary}

\newtheorem{lemma}[theorem]{Lemma}

\numberwithin{equation}{section}
\let\oldref\ref
\renewcommand{\ref}[1]{(\oldref{#1})}
\DeclareMathOperator{\Tr}{Tr}

\newcommand{\ket}[1]{{\vert{#1}\rangle}}
\newcommand{\bra}[1]{{\langle{#1}\vert}}

\newcommand{\wtil}{\widetilde}
\begin{document}

\title{Strong converse for the feedback-assisted classical capacity of entanglement-breaking channels} 
\author{Dawei Ding\thanks{Department of Applied Physics,   Stanford University, Stanford, California 94305-4090, USA}
\and Mark M. Wilde\thanks{Hearne Institute for Theoretical Physics, Department of Physics and Astronomy, Center for Computation and Technology, Louisiana State University,
Baton Rouge, Louisiana 70803, USA}}
\maketitle

\begin{abstract}
Quantum entanglement can be used in a communication scheme to establish a correlation between successive channel inputs that is impossible by classical means. It is known that the classical capacity of quantum channels can be enhanced by such entangled encoding schemes, but this is not always the case. In this paper, we prove that
a strong converse theorem holds for the classical capacity of an entanglement-breaking channel even when it is assisted
by a classical feedback link from the receiver to the transmitter. In doing so, we identify a bound on the strong converse exponent, which determines the exponentially decaying rate at which the success probability tends to zero, for a sequence of codes with  communication rate exceeding capacity. Proving a strong converse, along with an achievability theorem, shows that the classical capacity is a sharp boundary between reliable and unreliable communication regimes. One of the main tools in our proof is the sandwiched R\'enyi relative entropy. The same method of proof is used to derive an exponential bound on the success probability when communicating over an arbitrary quantum channel assisted by classical feedback, provided that the transmitter does not use entangled encoding schemes.
\end{abstract}

\section{Introduction}
The classical theory of communication is one of the modern successes of applied mathematics \cite{book1991cover,GK12}. It is arguably one of the foundations of our current information age and provides new ways of thinking about problems in many other fields of study, such as physics and in particular quantum mechanics. The interaction between these two fields is mutual; while some problems in quantum mechanics can be turned into communication problems, the existence of quantum phenomena strongly suggests that we should rethink many aspects of communication theory. Not only does the notion of a quantum state challenge what we mean by ``information,'' but the possibilities due to quantum mechanics give rise to new classes of communication protocols. With respect to this latter consideration, some fundamental motivating questions
for quantum information theory have traditionally been and still are the following:
\begin{enumerate}
  \item Is the theory of classical communication affected at a fundamental level by the consideration of quantum mechanical phenomena? 
  \item Does using quantum states and measurement for classical communication have practical advantages over using classical techniques?
\end{enumerate}

In an attempt to answer these questions, one of the primary
goals is to study the ability of a quantum channel to communicate classical information, that is, bits.
Like many communication problems,
this ability is quantified by the notion of channel capacity. The \textit{classical capacity} $C$ of a quantum channel $\mathcal{N}$ is defined to be the maximum rate of communication such that the decoding error probability can tend to zero in the limit of many channel uses. With this definition, one natural question is to determine
 how to compute the classical capacity. We know that the Holevo-Schumacher-Westmoreland (HSW) theorem \cite{H, SW} provides a lower bound:
\begin{equation}
  C(\mathcal{N}) \geq \chi(\mathcal{N}) \equiv \sup_{\{p_X(x),\rho_x\}} I(X;B)_\rho ,
\end{equation}
where $\{p_X(x),\rho_x\}$ is an ensemble of quantum states such that each $\rho_x$ can be input to the channel, and $I(X;B)_\rho \equiv H(X)_\rho + H(B)_\rho - H(XB)_\rho$ is the quantum mutual information of the following classical-quantum state:
\begin{equation}
  \rho_{XB} \equiv \sum_x p_X(x) |x\rangle \langle x|_X \otimes \mathcal{N}(\rho_x) ,
\end{equation}
where $\{|x\rangle\}$ is an orthonormal basis for the classical reference system and $H(G)_\sigma$ is the von Neumann entropy of a quantum state $\sigma_G$ on system $G$.
The quantity $\chi(\mathcal{N})$ is called the \textit{Holevo information} of the channel. 

The classical capacity $C(\mathcal{N})$ can actually be formally rewritten in terms of the Holevo information as well, via a procedure known as regularization. Before doing so, note that we obtain the HSW theorem by considering only encoding procedures that do not use entangled inputs\footnote{The highest possible rate with this restriction, proven by the HSW theorem to be $\chi(\mathcal{N})$, is then a lower bound on the classical capacity.}. That is, each quantum system sent through the channel is not entangled with any other system that is sent. To get the classical capacity, it is generally necessary to incorporate entangled inputs into the calculation.
The approach given by \cite{H, SW} is to multiplex the channel such that one use of this multiplexed channel corresponds to multiple uses of the original channel. A multipartite entangled state describing several inputs across different uses of the original channel can now be simulated by the corresponding single input state to the multiplexed channel. We can therefore express the maximum rate for blocks of size $n$ in terms of a Holevo information:
\begin{equation}
  \frac{1}{n} \chi(\mathcal{N}^{\otimes n})
  \label{multiplex}.
\end{equation}
To obtain the classical capacity, we simply allow for inputs entangled across arbitrarily many channel uses: 
\begin{equation}
  C(\mathcal{N}) = \lim_{n\rightarrow \infty} \frac{1}{n}\chi(\mathcal{N}^{\otimes n}) .
  \label{capLim}
\end{equation}
This idea put together with a converse theorem establishes
the regularized expression in \eqref{capLim} as being equal to the classical capacity.

Unfortunately, computing the classical capacity this way is clearly intractable. This prompts us to look for special cases. We observe that the limit in \eqref{capLim} is equal to $\chi(\mathcal{N})$ iff the Holevo information satisfies a tensor-power additivity property (see Appendix~\ref{app:additivity} for a brief derivation):
\begin{equation}
  \forall n \quad \chi(\mathcal{N}^{\otimes n}) = n \chi(\mathcal{N}).
  \label{add}
\end{equation}
Another way to restate the above is that entangled inputs do not increase capacity iff \eqref{add} is satisfied. However, this is not true for some channels, as was shown in \cite{notAdd}. The next question then appears: What characterizes channels that satisfy \eqref{add}?

A sufficient condition for a quantum channel to satisfy tensor-power additivity is for it to be entanglement-breaking (EB) \cite{EB}. Let $\mathcal{N}_{A\rightarrow B}$ denote a quantum channel, where the arrow notation indicates that the channel maps a state of the input system $A$ to a state in an output system $B$. An EB channel is defined such that for any bipartite state $\rho_{AA'}$, the output state
\begin{equation}
  \left(\mathcal{N}_{A\rightarrow B} \otimes \operatorname{id}_{A'}  \right) \left(\rho_{AA'}\right)
\end{equation}
is separable. That is, the output of an EB channel can be written as a convex sum of product states. Effectively, the channel ``breaks'' the entanglement between $A$ and $A'$. Previous results have established that the Holevo information of EB channels is additive\footnote{Shor proved that the Holevo information of a tensor product of an EB channel with any other channel is equal to the sum of their respective Holevo informations. This form of additivity is stronger than \eqref{add} and is the one usually found in the literature.} \cite{ShorEB}. This is intuitive since any entanglement of the inputs is broken by the channel. Following this line of thinking, one could consider making a stronger statement by allowing additional resources to assist the communication but not to the point that entanglement can be established. Indeed, \cite{EBFB} proves a generalization of tensor-power additivity for EB channels with noiseless classical feedback. Furthermore, their results show that classical feedback does not increase the capacity of EB channels.

There are also possible stronger statements in another direction. The direct part of the original HSW theorem states that if the rate $R$ of communication is less than the Holevo information $\chi$, then there exists a sequence of protocols $\mathcal{P}^n$ such that the probability of error for such a sequence satisfies 
\begin{equation}
  \lim_{n\rightarrow \infty} p_e(n) = 0 ,
  \label{}
\end{equation}
where $n$ is the number of channel uses. A result from \cite{PhysRevA.76.062301} sharpens this claim by showing that 
there exists a sequence $\mathcal{P}^n$ such that
\begin{equation}
  p_e(n) \leq 2^{-kn} ,
  \label{}
\end{equation}
for some $k>0$ determined by the channel. The converse part of the original HSW theorem can be strengthened in a similar manner. It states that if $R>C(\mathcal{N})$, for any sequence of protocols $\mathcal{P}^n$, then
\begin{equation}
  \lim_{n\rightarrow \infty} p_e(n) > 0 .
  \label{}
\end{equation}
This is known as the \textit{weak converse}. In contrast, a \textit{strong converse} is symmetric to the achievability result above and states that regardless of the protocols used, the success probability decreases to zero in the asymptotic limit whenever $R>C(\mathcal{N})$. The strong converse can be sharpened as well whenever there is a constant separation between $R$ and $C(\mathcal{N})$, such that the convergence of the success probability to zero is exponential in $n$.

There are many reasons why we would want to prove a strong converse. First, a strong converse enriches our understanding of the capacity. A strong converse along with an achievability theorem shows that the capacity is a sharp boundary between reliable and unreliable communication regimes. This, amongst other results, indicates that the classical capacity of a quantum channel is a fundamental quantity of interest. Second, a strong converse is more relevant in practice than is the weak converse. A realistic quantum communication scheme has a finite blocklength; that is, the encoding is across a finite number of channel uses. Although the weak converse does provide a lower bound on the probability of error in the non-asymptotic regime, the strong converse improves the bound considerably. It expresses a trade-off between rate, error probability, and blocklength restrictive enough to be easily checked numerically or experimentally.

While a classical version of the strong converse is known for arbitrary discrete memoryless classical channels \cite{WolfStrong, A73}, it is still open whether or not strong converses hold for memoryless quantum channels.
After some early work \cite{ON99,W}, the strong converse has been proved for special cases, in particular for channels with certain symmetry \cite{KW}, for EB channels \cite{WWY13}, and for a wide class of quantum Gaussian channels \cite{BPWW14}. Given the strong converse results for EB channels \cite{WWY13} and the weak converse for EB channels with feedback \cite{EBFB}, it is natural to ask if these two statements are true at the same time. We can also ask directly for the strong converse for unentangled inputs when a feedback link is available. These are the main questions that we address in this paper.

\section{Summary of Results}

In this paper, we derive an explicit exponential bound on the success probability of a classical communication scheme that uses an entanglement-breaking channel along with classical feedback. The same method of proof can be used to establish a bound for arbitrary quantum channels with classical feedback, provided that the inputs are not entangled across multiple uses of the channel. When the communication rate exceeds the classical capacity, these exponential bounds immediately imply strong converse theorems for these settings.

We now provide an outline of the proof:
\begin{enumerate}
  \item First, taking as a starting point the general approach of \cite{N}, relating hypothesis testing to unassisted communication, we bound the success probability of an arbitrary feedback-assisted classical communication protocol by a sandwiched R\'enyi relative entropy \cite{MDSFT13,WWY13}.
  \item Next, one of the main observations from \cite{EBFB} is that the sender and receiver's systems are separable at all times throughout such a protocol, whenever the communication channel is entanglement-breaking. We use this fact and an entropy inequality from \cite{King} to split the relative entropy into two terms. The first term is bounded by an $\alpha$-information radius, which is a measure of the range of states a channel can output. We then equate this to the sandwiched $\alpha$-Holevo information \cite{WWY13}, which is a R\'enyi generalization of the Holevo information. The second term is bounded via monotonicity by another sandwiched R\'enyi relative entropy, to which we recursively apply the same argument.
  \item This gives the following bound on the probability of success for any finite blocklength $n$:
    \begin{equation}
      p_{\operatorname{succ}} \leq 2^{-n\sup_{\alpha>1} \frac{\alpha-1}{\alpha} \left( R- \widetilde{\chi}_\alpha(\mathcal{N}) \right) } ,
      \label{tehEqn}
    \end{equation}
where $R$ is the rate of communication and $\mathcal{N}$ is the EB channel. It follows from previous arguments \cite{WWY13} that when $R>\chi(\mathcal{N})$, the right hand side of \eqref{tehEqn} is a decaying exponential, thereby establishing the strong converse. We provide an alternate (arguably simpler) proof of this fact (similar to those in \cite{MH,QFBStrong}) by establishing that the $\alpha$-Holevo information converges continuously to the Holevo information as $\alpha$ approaches $1$.
\end{enumerate}
 Appendix~\ref{app:bowen} includes a brief review of the argument for the weak converse from \cite{EBFB}.

\section{Preliminaries}
In this section we provide some necessary definitions, concepts, and previous results used in the derivation of \eqref{tehEqn}.

\subsection{Quantum states, measurements, operator norms, and quantum channels}
We start with definitions of relevant mathematical notions from quantum mechanics. Given a finite-dimensional Hilbert space $\mathcal{H}$, let $\mathcal{B}(\mathcal{H})$ denote the algebra consisting of linear operators acting on $\mathcal{H}$. A relevant measure of a operator $X$ is its \textit{Schatten $\alpha$-norm}, which is defined as
\begin{equation}
  \Vert X \Vert_\alpha \equiv \left\{ \Tr\left[ |X|^\alpha \right]\right\}^{1/\alpha},
\end{equation}
where $\alpha \geq 1$ and $|X| \equiv \sqrt{X^\dagger X}$.

The set of \textit{quantum states} is a convex subset of $\mathcal{B}(\mathcal{H})$ given by
\begin{equation}
  \mathcal{S}(\mathcal{H}) = \left\{ \rho \in \mathcal{B}(\mathcal{H}): \rho\geq 0, \Tr \rho=1 \right\} ,
\end{equation}
where the notation $\rho\geq0$ means that $\rho$ is positive semidefinite. For composite states, we consider the tensor product of two Hilbert spaces $\mathcal{H}_A$ and $\mathcal{H}_B$, denoted by $\mathcal{H}_A\otimes\mathcal{H}_B$. We can obtain from the overall density operator $\rho_{AB} \in \mathcal{S}(\mathcal{H}_A\otimes\mathcal{H}_B)$ the reduced density operator $\rho_A$ corresponding to the quantum state on only the $A$ system by performing a partial trace:
\begin{equation}
  \rho_A \equiv \Tr_B(\rho_{AB}). 
  \label{}
\end{equation}
The partial trace can be defined as
\begin{equation}
\Tr_B(\ket{x_1}\bra{x_2}_A \otimes \ket{y_1}\bra{y_2}_B) \equiv
\ket{x_1}\bra{x_2}_A \langle y_2 \vert y_1 \rangle
\end{equation}
for vectors 
$\ket{x_1}$, $\ket{x_2}$, 
$\ket{y_1}$, and $\ket{y_2}$, and then extended by linearity.
A state $\rho_{AB}$ is \textit{separable} if it can be written as
\begin{equation}
  \rho_{AB} = \sum_x p(x) \rho_A^x \otimes \rho_B^x,
  \label{}
\end{equation}
where $p(x)$ is a probability distribution and
$\{\rho_A^x\}$ and $\{\rho_B^x\}$ are sets of states.

A positive operator-valued measure (POVM) consists of a set $\left\{ \Lambda_m \right\}$ of positive semidefinite operators indexed by $m$ and corresponding to different measurement results. The set satisfies $\sum_m \Lambda_m =I$, which allows us to interpret the quantity
\begin{equation}
  p_m \equiv \Tr(\rho\Lambda_m)
\end{equation}
as the probability of measuring $m$ given a quantum state $\rho$.

We next consider maps on operators and in particular quantum states. A linear map $\Psi:\mathcal{B}(\mathcal{H}_A) \rightarrow \mathcal{B}(\mathcal{H}_B)$ is called \textit{positive} if $\forall \tau \in \mathcal{B}(\mathcal{H}_A)$, $\tau \geq 0$ implies $\Psi(\tau) \geq 0$. It is called \textit{completely positive} if id$_R \otimes \Psi$ is positive for an arbitrary auxiliary system $R$, where id$_R$ is the identity map on $R$. A map $\Psi$ is a \textit{quantum channel} if it is linear, completely positive, and also trace-preserving. 

\subsection{Classical feedback-assisted classical communication protocols}
\label{subsec:protocol}
\begin{figure}
[ptb]
\begin{center}
\includegraphics[width=\textwidth]%
{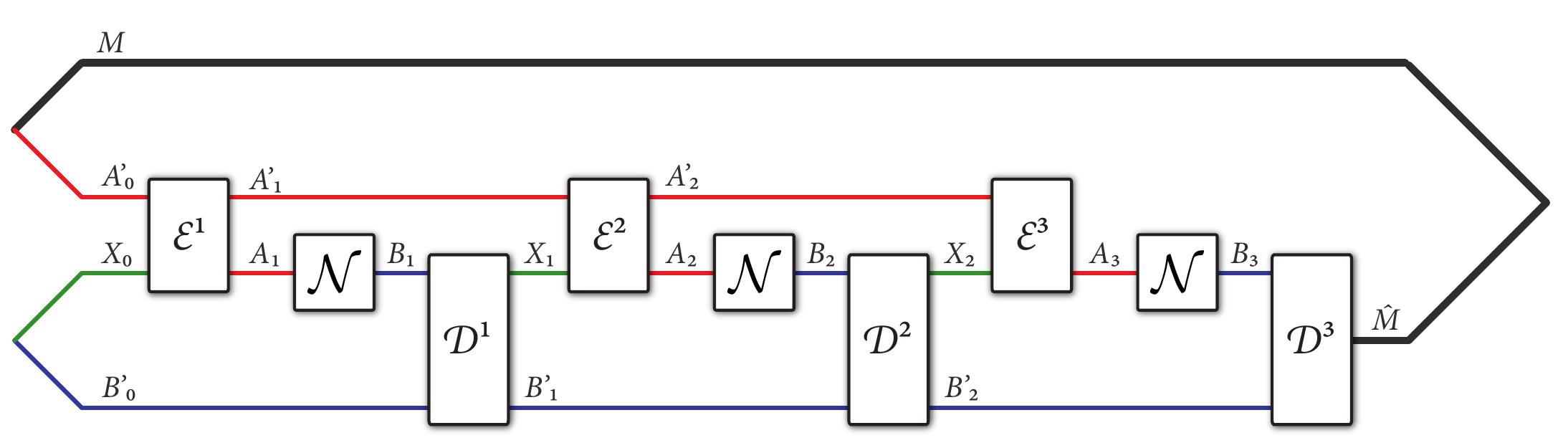}
\caption{A general protocol for feedback-assisted classical communication. The sender manipulates the systems labeled by $A$ and the receiver those labeled by $B$. Every $X$ system is classical and represents noiseless classical feedback from
the receiver to the sender.}%
\label{fig:FB-EB}%
\end{center}
\end{figure}
Figure \ref{fig:FB-EB} depicts the most general three-round protocol for classical feedback-assisted classical communication. The generalization to $n$ rounds is clear. The protocol begins with Alice preparing a classical register $M$ with the
message to be sent. This is correlated with some system $A_{0}^{\prime}$. Bob
uses the classical feedback channel to send a classical system $X_{0}$ correlated with some quantum system $B_{0}'$ to Alice. The global state
is then%
\begin{equation}
\rho_{MA_{0}^{\prime}X_{0}B_{0}'}\equiv\sum_{m}p_{M}(  m)
\left\vert m\right\rangle \left\langle m\right\vert _{M}\otimes\rho
_{A_{0}^{\prime}}^{m}\otimes\sum_{x}p_{X_0}(  x)  \left\vert
x\right\rangle \left\langle x\right\vert _{X_0}\otimes\rho_{B_{0}'}^{x}.
\label{initial}
\end{equation}
We will track the state conditioned on a particular value $m$ of the message
register $M$, given by%
\begin{equation}
\rho_{A_{0}^{\prime}X_{0}B_{0}'}^{m}\equiv\rho_{A_{0}^{\prime}}^{m}%
\otimes\sum_{x}p_{X_0}(  x)  \left\vert x\right\rangle \left\langle
x\right\vert _{X_0}\otimes\rho_{B_{0}'}^{x} .
\end{equation}
Alice performs an encoding $\mathcal{E}_{A_{0}^{\prime}X_{0}\rightarrow
A_{1}^{\prime}A_{1}}^{1}$. The state at this point is%
\begin{align}
\rho_{A_{1}^{\prime}A_{1}B_{0}'}^{m} &  \equiv\mathcal{E}_{A_{0}^{\prime
}X_{0}\rightarrow A_{1}^{\prime}A_{1}}^{1}\!\left(  \rho_{A_{0}^{\prime}%
X_{0}B_{0}'}^{m}\right)  \\
&  =\sum_{x}p_{X_0}(  x)  \mathcal{E}_{A_{0}^{\prime}X_{0}\rightarrow A_{1}^{\prime
}A_{1}}^{1}\!\left(  \rho_{A_{0}^{\prime}}^{m}\otimes\left\vert
x\right\rangle \left\langle x\right\vert _{X_0}\right)   
\otimes\rho_{B_{0}'}^{x},
\end{align}
We note that the state $\rho_{A_{1}^{\prime}A_{1}B_{0}'}^{m}$ is separable with respect
to the cut $A_{1}^{\prime}A_{1}:B_{0}'$. 

Next, Alice uses the EB channel $\mathcal{N}_{A_{1}\rightarrow
B_{1}}$ for the first time by sending system $A_{1}$, leading to the state%
\begin{equation}
\rho_{A_{1}^{\prime}B_{1}B_{0}'}^{m}\equiv\mathcal{N}_{A_{1}\rightarrow
B_{1}}\!\left(  \rho_{A_{1}^{\prime}A_{1}B_{0}'}^{m}\right)  .
\end{equation}
Since the channel is entanglement-breaking and the state before the channel was already separable with respect to the $A_{1}^{\prime}A_{1}:B_{0}'$ cut, the state after the channel is fully separable, that is, it is separable across all possible partitions. Now, consider instead that Alice uses an arbitrary channel $\mathcal{N}_{A_{1}\rightarrow B_{1}}$  but with an input separable across the $A_1' : A_1$ cut (such that $\rho_{A_1' A_1 B_0'}^m$ is fully separable). Then, the state $\rho_{A_{1}^{\prime}B_{1}B_{0}'}^{m}$ is again fully separable. We write it in the form
\begin{equation}
\rho_{A_{1}^{\prime}B_{1}B_{0}'}^{m}=\sum_{y}p_{Y}(
y)  \rho_{A_{1}^{\prime}}^{y,m}\otimes\rho_{B_{1}}^{y,m}%
\otimes\rho_{B_{0}'}^{y,m}.
\end{equation}

Bob now applies the decoding map $\mathcal{D}_{B_{1}B_{0}'\rightarrow
X_1B_{1}^{\prime}}^{1}$, where $X_1$ is the classical system that is sent back to Alice.
The state at this point is%
\begin{align}
\rho_{A_{1}^{\prime}X_1B_{1}^{\prime}}^{m} &  \equiv\mathcal{D}_{B_{1}%
B_{0}'\rightarrow X_1B_{1}^{\prime}}^{1}\left(  \rho_{A_{1}^{\prime}B_{1}B_{0}'%
}^{m}\right)  \\
&  =\sum_{x_{1}}p_{X_{1}}(  x_{1})  \left\vert x_{1}\right\rangle
\left\langle x_{1}\right\vert _{X_1}\otimes\sum_{z_{1}}p_{Z_{1}|X_{1}}(
z_{1}|x_{1})  \rho_{A_{1}^{\prime}}^{x_{1},z_{1}}\otimes\rho
_{B_{1}^{\prime}}^{x_{1},z_{1}}.
\end{align}
which is fully separable. Hence, after Alice applies the second encoder $\mathcal{E}^2$ and then sends it through the channel, the state will still be fully separable. The key observation here is that if we use an EB channel or an arbitrary channel with separable inputs, \textit{the state is always separable across a cut that divides Alice and Bob's systems.}

The only difference in subsequent rounds is the final measurement. Say there are $n$ rounds in the protocol. At the last round, Bob measures the state $\rho_{B_n B_{n-1}'}$ using a POVM given by $\left\{ D^m \right\}$ with elements corresponding to different possible messages that Alice sent.

\subsection{R\'enyi relative entropies and bounds on success probability}

An important classical information theoretic quantity is the R\'enyi relative entropy, which can be generalized to the quantum case in a number of ways. In this paper, we use the \textit{sandwiched quantum R\'enyi relative entropy} \cite{MDSFT13, WWY13} which is given by
\begin{equation}
  \widetilde{D}_\alpha(\rho \Vert \sigma) \equiv \left\{ 
    \begin{array}{ll}
      \frac{1}{\alpha-1} \log\left[ \Tr\left( \left( \sigma^{(1-\alpha)/(2\alpha)}\rho\sigma^{(1-\alpha)/(2\alpha)} \right)^\alpha \right)\right] & : \rho \not\perp \sigma \wedge (\text{supp}(\rho) \subseteq \text{supp}(\sigma) \vee \alpha \in (0,1)) \\
      +\infty & : \text{otherwise}
    \end{array}
    \right.
\end{equation}
where $\alpha \in (0,1) \cup (1,\infty)$ and $\rho \not\perp \sigma$ means $\rho,\sigma$ are non-orthogonal quantum states. Note that all logarithms in this paper are taken base two.

We now recall some properties of the sandwiched R\'enyi relative entropy. For fixed $\rho$ and $\sigma$, the function $\alpha\mapsto\widetilde{D}_\alpha(\rho \Vert \sigma)$ is monotone non-decreasing \cite{MDSFT13}. It also converges to the
quantum relative entropy $D(\rho\Vert\sigma)$ \cite{Umegaki} in the limit as $\alpha \to 1$ \cite{MDSFT13,WWY13}:
\begin{equation}
  \lim_{\alpha\rightarrow 1} \widetilde{D}_\alpha(\rho \Vert \sigma) = D(\rho\Vert\sigma) , \label{eq:sandwiched-to-Umegaki}
\end{equation}
where
\begin{equation}
D(\rho\Vert\sigma) \equiv \left\{ 
    \begin{array}{ll}
      \Tr\left[ \rho\left( \log\rho-\log\sigma \right) \right] & :
      \text{supp}(\rho) \subseteq \text{supp}(\sigma)\\
      +\infty & : \text{otherwise}
    \end{array}
    \right. .
\end{equation}
 Furthermore, it satisfies the data-processing inequality for $\alpha \in [1/2,1)\cup (1,\infty)$ \cite{FL13,B13monotone}; that is, for all quantum channels $\mathcal{N}$,
\begin{equation}
   \widetilde{D}_\alpha(\mathcal{N}(\rho) \Vert \mathcal{N}(\sigma)) \leq \widetilde{D}_\alpha(\rho \Vert \sigma) .
\end{equation}
In particular, consider the following replacement channel which simply replaces the input with some state $\omega$:
\begin{equation}
 \mathcal{R_\omega}(\rho) \equiv \Tr(\rho) \omega .
\end{equation}
It is easy to see that for all $\omega$, $\widetilde{D}_\alpha(\omega\Vert\omega) =0$, implying that
\begin{equation}
    \widetilde{D}_\alpha(\rho \Vert \sigma)\geq\widetilde{D}_\alpha(\mathcal{R_\omega}(\rho) \Vert \mathcal{R_\omega}(\sigma)) =0 ,
\end{equation}
which shows that the sandwiched R\'enyi relative entropy is non-negative
whenever its arguments are quantum states $\rho$ and $\sigma$.

Using the sandwiched R\'enyi relative entropy, we can define the
sandwiched \textit{$\alpha$-Holevo information} \cite{WWY13} of an ensemble, that is, a classical probability distribution of quantum states, $\left\{ p_X(x),\rho_x \right\}$ as
\begin{equation}
  \widetilde{\chi}_\alpha(\left\{ p_X(x),\rho_x \right\}) \equiv \inf_{\sigma_R \in \mathcal{S}(\mathcal{H}_R)} \widetilde{D}_\alpha(\rho_{XR}\Vert\rho_X\otimes\sigma_R)
\end{equation}
where
\begin{equation}
  \rho_{XR} \equiv \sum_x p_X(x) \vert x\rangle\langle x\vert_X \otimes \left( \rho_x \right)_R
\end{equation}
and $\rho_x$ are states of a system $R$. With this, we define the $\alpha$-Holevo information of a quantum channel $\mathcal{N}$ as
\begin{equation}
  \widetilde{\chi}_\alpha(\mathcal{N}) \equiv \sup_{\left\{p_X(x),\rho_x \right\}} \widetilde{\chi}_\alpha\left(\left\{ p_X(x),\mathcal{N}(\rho_x) \right\}\right) .
\end{equation}
We also define the $\alpha$-information radius of a channel $\mathcal{N}:
\mathcal{S}(\mathcal{H}_A) \mapsto
\mathcal{S}(\mathcal{H}_B)$ as
\begin{equation}
  \widetilde{K}_\alpha(\mathcal{N}) \equiv \inf_{\sigma \in \mathcal{S}(\mathcal{H}_B)} \sup_{\rho \in \mathcal{S}(\mathcal{H}_A)} \widetilde{D}_\alpha (\mathcal{N}(\rho) \Vert \sigma)
  \label{} .
\end{equation}
We note that since $\widetilde{D}_\alpha$ is monotonically non-decreasing in $\alpha$, so are $\widetilde{\chi}_\alpha$ and $\widetilde{K}_\alpha$.

\section{Strong converse}
This section is dedicated to proving the main theorem of this paper. We will need the following lemmas. The first is proven using an inequality from \cite{N} in Lemma 5 of \cite{QFBStrong} via monotonicity of $\widetilde{D}_\alpha$. The second will be used to take advantage of the separability of the quantum state observed in Section \ref{subsec:protocol}. The third is an equality between $\alpha$-Holevo information and $\alpha$-information radius. The fourth states that $\alpha$-Holevo information and $\alpha$-information radius respectively tend to the conventional Holevo information and information radius in the limit $\alpha \to 1$, which we prove in Appendix \ref{sec:pfLemma}. Note that this establishes Lemma~\ref{chiK} as a generalization of the equality $\chi(\mathcal{N}) = K(\mathcal{N})$ \cite{OPW,SBW,SW02}.
\begin{lemma} 
   \label{Nagaoka}
  Let $\alpha>1$, $\rho,\sigma\in \mathcal{S}(\mathcal{H})$, and $\Lambda$ be such that $0 \leq \Lambda \leq I$. Let
  \begin{equation}
    p \equiv \Tr\left[ \Lambda \rho \right],
    \,\,\, q\equiv \Tr\left[ \Lambda\sigma \right].
  \end{equation}
  Then
  \begin{equation}
    \widetilde{D}_\alpha(\rho\Vert\sigma) \geq \frac{1}{\alpha-1}\log\left[ p^\alpha q^{1-\alpha} \right] .
  \end{equation} 
\end{lemma}
\begin{lemma}
[\cite{King,H06}] \label{King}
Let $P_{AB}$ be a positive semidefinite separable operator. Such an operator can be written in the following form:
\begin{equation}
P_{AB}=\sum_{j}C_{A}^{j}\otimes D_{B}^{j},
\end{equation}
where $C_{A}^{j},D_{B}^{j}\geq0$ for all $j$. Let $P_{B}=\operatorname{Tr}%
_{A}\left\{  P_{AB}\right\}  $ and let $\mathcal{M}_A$ be a completely positive
linear map acting on the $A$ system. Then, for all $\alpha\geq1$,%
\begin{equation}
\left\Vert \left(  \mathcal{M}_{A}\otimes\operatorname{id}_{B}\right)  \left(
P_{AB}\right)  \right\Vert _{\alpha}\leq\nu_{\alpha}(  \mathcal{M}%
_{A})  \cdot\left\Vert P_{B}\right\Vert _{\alpha},
\end{equation}
where $\nu_{\alpha}(  \mathcal{M}_{A})  $ is the $1\rightarrow
\alpha$ norm of $\mathcal{M}_A$, defined as%
\begin{equation}
  \nu_{\alpha}(  \mathcal{M}_A)  \equiv\sup_{X\neq0, X \in \mathcal{B}(\mathcal{H}_A)}\frac{\left\Vert
\mathcal{M}_A\left(  X\right)  \right\Vert _{\alpha}}{\left\Vert X\right\Vert
_{1}}.
\end{equation}
\end{lemma}
\begin{lemma}
  [\cite{WWY13}]  \label{chiK}
  For $\alpha > 1$, the $\alpha$-Holevo information and the $\alpha$-information radius are the same:
  \begin{equation}
    \widetilde{\chi}_\alpha(\mathcal{N}) = \widetilde{K}_\alpha(\mathcal{N})
    \label{}.
  \end{equation}
\end{lemma}

\begin{proof}
This statement was essentially proved as Lemma 14 of 
\cite{WWY13}, but only for the interval $\alpha \in (1,2]$. To get the statement for all $\alpha>1$, we actually need to extend Lemma 14 in \cite{WWY13} slightly. This follows from the proof given there and the observations that for all $\alpha>1$, $x^{\left(  1-\alpha\right)/\alpha}$ is operator convex, Tr$\left\{  x^{\alpha}\right\}$ is convex, and $\widetilde{D}_{\alpha}$ is jointly quasi-convex.

For completeness and convenience, we also give a full proof here, following the steps in the proof of Lemma 14 in \cite{WWY13} closely. We first prove the inequality $\widetilde{K}_{\alpha}\left(  \mathcal{N}%
\right)  \leq\widetilde{\chi}_{\alpha}\left(  \mathcal{N}\right)  $
for $\alpha>1$. Defining  $\widetilde{Q}_{\alpha}(\rho\Vert \sigma) \equiv \Tr((\sigma^{(1-\alpha)/2\alpha}\rho \sigma^{(1-\alpha)/2\alpha})^\alpha)$, consider that%
\begin{align}
\widetilde{K}_{\alpha}\left(  \mathcal{N}\right)    & =\inf_{\sigma}\sup
_{\rho}\widetilde{D}_{\alpha}\left(  \mathcal{N}(  \rho)
\Vert\sigma\right)  \\
& =\inf_{\sigma}\sup_{\rho}\frac{1}{\alpha-1}\log\widetilde{Q}_{\alpha}\left(
\mathcal{N}(  \rho)  \Vert\sigma\right)  \\
& =\frac{1}{\alpha-1}\log\inf_{\sigma}\sup_{\rho}\widetilde{Q}_{\alpha}\left(
\mathcal{N}(  \rho)  \Vert\sigma\right)
\end{align}
So now we focus on the $\widetilde{Q}_{\alpha}$ quantity and find that%
\begin{align}
\inf_{\sigma}\sup_{\rho}\widetilde{Q}_{\alpha}\left(  \mathcal{N}\left(
\rho\right)  \Vert\sigma\right)    & \leq\inf_{\sigma}\sup_{\mu}\int
d\mu(  \rho)  \ \widetilde{Q}_{\alpha}\left(  \mathcal{N}\left(
\rho\right)  \Vert\sigma\right)   \\
& =\sup_{\mu}\inf_{\sigma}\int d\mu(  \rho)  \ \widetilde
{Q}_{\alpha}\left(  \mathcal{N}(  \rho)  \Vert\sigma\right)  \\
& =\sup_{\left\{  p_{X}(  x)  ,\rho_{x}\right\}  }\inf_{\sigma}%
\sum_{x}p_{X}(  x)  \widetilde{Q}_{\alpha}\left(  \mathcal{N}%
\left(  \rho_{x}\right)  \Vert\sigma\right)  \\
& =\sup_{\left\{  p_{X}(  x)  ,\rho_{x}\right\}  }\inf_{\sigma_{B}%
}\widetilde{Q}_{\alpha}\left(  \rho_{XB}\Vert\rho_{X}\otimes\sigma_{B}\right)
\label{eq:K<=chi}
\end{align}
The first inequality follows by taking a supremum over all probability
measures $\mu$ on the set of all states $\rho$. The first equality is a result of
applying the Sion minimax theorem \cite{S58}---we can do so because the function $\int
d\mu(  \rho)  \ \widetilde{Q}_{\alpha}\left(  \mathcal{N}\left(
\rho\right)  \Vert\sigma\right)  $ is linear in the probability measure $\mu$
and convex in states $\sigma$. Convexity of $\widetilde{Q}_{\alpha}\left(
\mathcal{N}(  \rho)  \Vert\sigma\right)  $ in $\sigma$ follows
because
\begin{equation}
\widetilde{Q}_{\alpha}\left(  \mathcal{N}(  \rho)  \Vert
\sigma\right)  =\text{Tr}\left\{  \left(  \left[  \mathcal{N}\left(
\rho\right)  \right]  ^{1/2}\sigma^{\left(  1-\alpha\right)  /\alpha}\left[
\mathcal{N}(  \rho)  \right]  ^{1/2}\right)  ^{\alpha}\right\}  ,
\end{equation}
$x^{\left(  1-\alpha\right)  /\alpha}$ is operator convex for $\alpha >1$
and $\Tr(x^{\alpha})$ is convex for $\alpha > 1$. The second equality
follows by an application of the Fenchel-Eggleston-Caratheodory theorem
(see \cite{EK12}, for example):\ the
function $\widetilde{Q}_{\alpha}\left(  \mathcal{N}(  \rho)
\Vert\sigma\right)  $ is continuous in $\rho$, which is a density operator
acting on a $d$-dimensional Hilbert space, so that for each $\mu$, there exists
a probability distribution $p_{X}(  x)  $ on no more than $d^{2}$
letters such that%
\begin{equation}
\int d\mu(  \rho)  \ \widetilde{Q}_{\alpha}\left(  \mathcal{N}%
(  \rho)  \Vert\sigma\right)  =\sum_{x}p_{X}(  x)
\widetilde{Q}_{\alpha}\left(  \mathcal{N}\left(  \rho_{x}\right)  \Vert
\sigma\right)  .
\end{equation}
The last equality in \eqref{eq:K<=chi} follows from the properties of $\widetilde{Q}_{\alpha
}$ and by defining%
\begin{equation}
\rho_{XB}\equiv\sum_{x}p_{X}(  x)  \left\vert x\right\rangle
\left\langle x\right\vert _{X}\otimes\left[  \mathcal{N}\left(  \rho
_{x}\right)  \right]  _{B}.
\end{equation}
So we can then conclude that%
\begin{align}
\widetilde{K}_{\alpha}\left(  \mathcal{N}\right)    & \leq\frac{1}{\alpha
-1}\log\sup_{\left\{  p_{X}(  x)  ,\rho_{x}\right\}  }\inf
_{\sigma_{B}}\widetilde{Q}_{\alpha}\left(  \rho_{XB}\Vert\rho_{X}\otimes
\sigma_{B}\right)  \\
& =\sup_{\left\{  p_{X}(  x)  ,\rho_{x}\right\}  }\inf_{\sigma_{B}%
}\frac{1}{\alpha-1}\log\widetilde{Q}_{\alpha}\left(  \rho_{XB}\Vert\rho
_{X}\otimes\sigma_{B}\right)  \\
& =\widetilde{\chi}_{\alpha}\left(  \mathcal{N}\right)  .
\end{align}

The proof of the other inequality $\widetilde{K}_{\alpha}\left(
\mathcal{N}\right)  \geq\widetilde{\chi}_{\alpha}\left(  \mathcal{N}\right)
$\ is simpler. Consider that
\begin{align}
\widetilde{\chi}_{\alpha}\left(  \mathcal{N}\right)   &  =\sup_{\left\{
p_{X}(  x)  ,\rho_{x}\right\}  }\inf_{\sigma_B}\widetilde{D}_{\alpha
}\left(  \rho_{XB}\Vert\rho_{X}\otimes\sigma_B\right)  \\
&  \leq\sup_{\left\{  p_{X}(  x)  ,\rho_{x}\right\}  }\widetilde
{D}_{\alpha}\left(  \rho_{XB}\Vert\rho_{X}\otimes\sigma\right)  \\
&  \leq\sup_{\left\{  p_{X}(  x)  ,\rho_{x}\right\}  }\sup_{x}\widetilde{D}_{\alpha}\left(  \left\vert x\right\rangle
\left\langle x\right\vert \otimes\mathcal{N}\left(  \rho_{x}\right)
\Vert\left\vert x\right\rangle \left\langle x\right\vert \otimes\sigma\right)
\\
&  =\sup_{\left\{  p_{X}(  x)  ,\rho_{x}\right\}  }\sup_{x}\widetilde{D}_{\alpha}\left(  \mathcal{N}\left(  \rho_{x}\right)
\Vert\sigma\right)  \\
&  \leq\sup_{\rho}\widetilde{D}_{\alpha}\left(  \mathcal{N}\left(
\rho\right)  \Vert\sigma\right)  .
\end{align}
The second inequality follows from joint quasi-convexity of $\widetilde
{D}_{\alpha}$ for $\alpha > 1$ \cite{FL13,B13monotone}, where a function $f$ is jointly quasi-convex if
\begin{equation}
  f(\lambda x_1 + (1-\lambda) x_2, \lambda y_1 + (1-\lambda) y_2) \le \max\left\{ f(x_1, y_1), f(x_2, y_2) \right\}.
\end{equation}
Since the above
inequality holds for all states $\sigma$, we can conclude that $\widetilde
{K}_{\alpha}\left(  \mathcal{N}\right)  \geq\widetilde{\chi}_{\alpha}\left(
\mathcal{N}\right)  $. 
\end{proof}

\begin{lemma}  \label{chiLims}
  For a quantum channel $\mathcal{N}$, the following limits hold:
  \begin{equation}
    \lim_{\alpha\rightarrow 1} \widetilde{\chi}_\alpha(\mathcal{N}) = \chi(\mathcal{N})
    \label{}
  \end{equation}
  and
  \begin{equation}
    \lim_{\alpha\rightarrow 1} \widetilde{K}_\alpha(\mathcal{N}) = K(\mathcal{N}),
    \label{}
  \end{equation}
  where
  \begin{equation}
    K(\mathcal{N}) \equiv \inf_{\sigma \in \mathcal{S}(\mathcal{H})} \sup_{\rho \in \mathcal{S}(\mathcal{H})} D (\mathcal{N}(\rho) \Vert \sigma)
    \label{}
  \end{equation}
  is the information radius.
\end{lemma}

\noindent We now state the theorem.
\begin{theorem}  \label{tehThm}
  Given any $n$-round protocol for classical feedback-assisted classical communication over an entanglement-breaking channel $\mathcal{N}$ with rate $R$, the average probability of success is bounded from above by an exponential in $n$:
  \begin{equation}
    p_{\operatorname{succ}} \leq 2^{ -n\sup_{\alpha>1} \frac{\alpha-1}{\alpha} \left( R- \widetilde{\chi}_\alpha(\mathcal{N}) \right) } , \label{eq:main-bound-1}
  \end{equation}
where $\widetilde{\chi}_\alpha$ is the $\alpha$-Holevo information. The same bound holds for an arbitrary channel $\mathcal{N}$ given that the encoder does not entangle inputs across different uses of the channel.
\end{theorem}

\begin{proof}
  We take as a starting point the approach of Nagaoka \cite{N}, connecting hypothesis testing with data processing of a R\'enyi information quantity. Let $\mathcal{P}_n$ be such a protocol. We are bounding the average probability of success, so we assume Alice chooses her messages uniformly at random. Using the notation of Section \ref{subsec:protocol}, the state we have at the final round of the protocol $\mathcal{P}_n$ is
\begin{equation}
  \rho_{MB_n B_{n-1}'} = \frac{1}{L}\sum_{m=1}^{L} \vert m\rangle\langle m\vert_M \otimes \rho_{B_n B_{n-1}'}^m ,
\end{equation}
where $L$ is the number of possible messages. 

Following the argument in \cite{QFBStrong}, we can write the success probability as
\begin{equation}
  p_{\operatorname{succ}} =
  \frac{1}{L}\sum_m \Tr\left[ D^m_{B_n B_{n-1}'} \rho_{B_n B_{n-1}'}^m \right] = \Tr\left[ T_{M B_n B_{n-1}'}\rho_{M B_n B_{n-1}'} \right] ,
\end{equation}
where
\begin{equation}
  T_{M B_n B_{n-1}'} \equiv \sum_m \vert m\rangle\langle m\vert_M \otimes D^m_{ B_n B_{n-1}'}.
\end{equation}
Note that $0 \leq T \leq I$, so that $\left\{ T, I-T \right\}$ is a POVM.

We now consider the state $\tau_{MB_n B_{n-1}'}$ defined to be the final state if we had implemented $\mathcal{P}_n$ using a
replacement channel $\mathcal{R_\sigma}$ instead of the original channel (see Figure~\ref{fig:FB-EB-replacement}).
\begin{figure}
[ptb]
\begin{center}
\includegraphics[width=\textwidth]%
{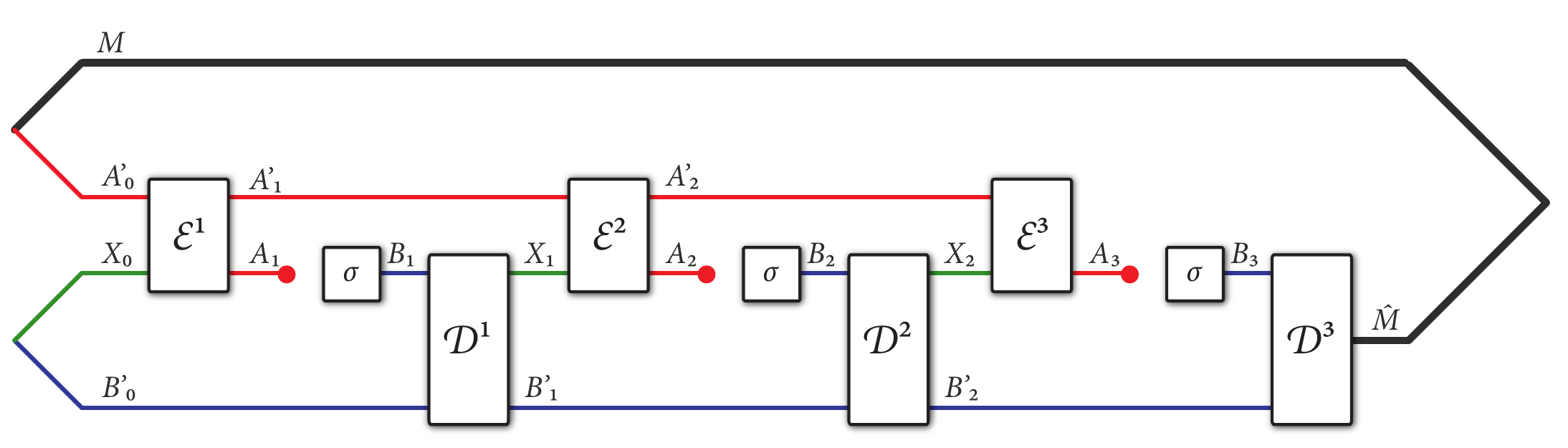}
\caption{A protocol for feedback-assisted classical communication when using the replacement channel. Notice that the communication line between $A_1$ and $\hat{M}$ is broken.}%
\label{fig:FB-EB-replacement}%
\end{center}
\end{figure}
The overall state in this alternate scenario has a simple expression since Bob's states are now independent
 of~$m$:
\begin{equation}
  \tau_{M B_n B_{n-1}'} = \tau_M \otimes \sigma_{B_n} \otimes \tau_{B_{n-1}'} ,
\end{equation}
where $\tau_M = I_M / L$. In the above, we remind that $B_n$ is a label for the output system of the $n$th replacement channel $\mathcal{R_\sigma}$, and the state of this system is equal to $\sigma$. We can then compute the following:
\begin{equation}
  \Tr\left[ T_{M B_n B_{n-1}'} \tau_{M B_n B_{n-1}'} \right] = \frac{1}{L} \sum_m \Tr\left[ D^m_{B_n B_{n-1}'} \sigma_{B_n} \otimes \tau_{B_{n-1}'} \right] = \frac{1}{L} ,
\end{equation}
where the last equality follows because $\sum_m D^m = I$ and $\Tr[\sigma_{B_n} \otimes \tau_{B_{n-1}'}]=1$. This equality is intuitive, since for a replacement channel, the receiver Bob cannot do any better than to guess the input message $m$ at random. Letting $T$ be the operator $\Lambda$ in Lemma \ref{Nagaoka}, we conclude that for $\alpha>1$,
\begin{equation}
  \widetilde{D}_\alpha(\rho_{M B_n B_{n-1}'} \Vert \tau_{M B_n B_{n-1}'}) \geq \frac{1}{\alpha-1} \log\left[
  p_{\operatorname{succ}}^\alpha \frac{1}{L^{1-\alpha}} \right], 
\end{equation}
which can be re-written as follows:
\begin{equation}
  \frac{\alpha}{\alpha-1}\log(p_{\operatorname{succ}}) + \log L \leq\widetilde{D}_\alpha(\rho_{M B_n B_{n-1}'} \Vert \tau_{M B_n B_{n-1}'}).
\end{equation}
We note that if the support of the state $\sigma$ does not contain the support of the image of the channel $\mathcal{N}_{A\rightarrow B}$, then the
$\widetilde{D}_\alpha$
upper bound is trivially equal to $+\infty$. So, in what follows, we choose $\sigma$ such that its support contains the support of the image of the channel in order to guarantee that all quantities involved in the forthcoming proof are finite.

We would therefore like to bound $\widetilde{D}_\alpha(\rho_{M B_n B_{n-1}'} \Vert \tau_{M B_n B_{n-1}'})$. To do so, consider that
\begin{multline}
\widetilde{D}_{\alpha}(\rho_{MB_{n}B_{n-1}^{\prime}}\Vert\tau_{MB_{n}%
B_{n-1}^{\prime}})   =\widetilde{D}_{\alpha}(\mathcal{N}_{A_{n}\rightarrow
B_{n}}(  \rho_{MA_{n}B_{n-1}^{\prime}})  \Vert\tau_{MB_{n-1}%
^{\prime}}\otimes\sigma_{B_{n}})\\
  =\frac{\alpha}{\alpha-1}\log\left\Vert \left(  \Theta_{\sigma_{B_{n}%
}^{\frac{1-\alpha}{\alpha}}}\circ\mathcal{N}_{A_{n}\rightarrow B_{n}}\right)
\left(  \tau_{MB_{n-1}^{\prime}}^{\left(  1-\alpha\right)  /(2\alpha)}%
\rho_{MA_{n}B_{n-1}^{\prime}}\tau_{MB_{n-1}^{\prime}}^{\left(  1-\alpha
\right)  /(2\alpha)}\right)  \right\Vert _{\alpha}
\label{normBound} ,
\end{multline}
where $\circ$ denotes function composition and in the last equality, we define
$\Theta$ by
\begin{equation}
  \Theta_\sigma(\rho) \equiv \sigma^{1/2} \rho \sigma^{1/2},
\end{equation}
 and the identity operation on the other systems is implied.
We now use the key observation from Section \ref{subsec:protocol} (and used in \cite{EBFB}) that if $\mathcal{N}$ is EB or if Alice uses separable inputs, throughout $\mathcal{P}_n$, Alice and Bob's systems are always separable. Furthermore, the $M$ system is classical, so $\rho_{MA_{n}B_{n-1}%
^{\prime}}$ is separable with respect to the $A_{n}:MB_{n-1}^{\prime}$ cut. This implies that
\begin{equation}
\tau_{MB_{n-1}^{\prime}}^{\left(  1-\alpha\right)  /(2\alpha)}\rho_{MA_{n}B_{n-1}^{\prime}}\tau_{MB_{n-1}^{\prime}}^{\left(  1-\alpha\right)/(2\alpha)}
\end{equation}
 is a positive semidefinite separable operator:
\begin{align}
\tau_{MB_{n-1}^{\prime}}^{\left(  1-\alpha\right)  /(2\alpha)}\rho
_{MA_{n}B_{n-1}^{\prime}}\tau_{MB_{n-1}^{\prime}}^{\left(  1-\alpha\right)
/(2\alpha)} &  =\tau_{MB_{n-1}^{\prime}}^{\left(  1-\alpha\right)  /(2\alpha)
}\left(  \sum_{j}p(j)\rho_{A_{n}}^{j}\otimes\rho_{MB_{n-1}^{\prime}}%
^{j}\right)  \tau_{MB_{n-1}^{\prime}}^{\left(  1-\alpha\right)  /(2\alpha)}\\
&  =\sum_{j}p(j)\rho_{A_{n}}^{j}\otimes\left(  \tau_{MB_{n-1}^{\prime}%
}^{\left(  1-\alpha\right)  /(2\alpha)}\rho_{MB_{n-1}^{\prime}}^{j}%
\tau_{MB_{n-1}^{\prime}}^{\left(  1-\alpha\right)  /(2\alpha)}\right).
\end{align}
Since conjugation by a positive semidefinite operator is clearly a completely positive map, we can apply Lemma \ref{King} to conclude that
\begin{multline}
 \left\Vert \left(  \Theta_{\sigma_{B_{n}}^{\frac{1-\alpha}{\alpha}}}%
\circ\mathcal{N}_{A_{n}\rightarrow B_{n}}\right)  \left(  \tau_{MB_{n-1}%
^{\prime}}^{\left(  1-\alpha\right)  /(2\alpha)}\rho_{MA_{n}B_{n-1}^{\prime}%
}\tau_{MB_{n-1}^{\prime}}^{\left(  1-\alpha\right)  /(2\alpha)}\right)
\right\Vert _{\alpha}\\
 \leq\nu_{\alpha}\left(  \Theta_{\sigma_{B_{n}}^{\frac{1-\alpha}{\alpha}}}%
\circ\mathcal{N}_{A_{n}\rightarrow B_{n}}\right)  \cdot\left\Vert
\tau_{MB_{n-1}^{\prime}}^{\left(  1-\alpha\right)  /(2\alpha)}\rho
_{MB_{n-1}^{\prime}}\tau_{MB_{n-1}^{\prime}}^{\left(  1-\alpha\right)
/(2\alpha)}\right\Vert _{\alpha}  \label{eq:helper1} .
\end{multline}
We then have the following chain of inequalities:
\begin{align}
\widetilde{D}_{\alpha}(\rho_{MB_{n}B_{n-1}^{\prime}}\Vert\tau_{MB_{n}%
B_{n-1}^{\prime}}) 
& \leq\frac{\alpha}{\alpha-1}\log\nu_{\alpha}\!\left(
\Theta_{\sigma^{\frac{1-\alpha}{\alpha}}}\circ\mathcal{N}\right)  +\widetilde{D}_{\alpha}(\rho_{MB_{n-1}%
^{\prime}}\Vert\tau_{MB_{n-1}^{\prime}})\\
& \leq \frac{\alpha}{\alpha-1}\log\nu_{\alpha}\!\left(
\Theta_{\sigma^{\frac{1-\alpha}{\alpha}}}\circ\mathcal{N}\right)+\widetilde{D}_{\alpha}(\rho_{MB_{n-1}B_{n-2}^{\prime}}\Vert
\tau_{MB_{n-1}B_{n-2}^{\prime}})\\
& \leq n\frac{\alpha}{\alpha-1}\log\nu_{\alpha}\!\left(
\Theta_{\sigma^{\frac{1-\alpha}{\alpha}}}\circ\mathcal{N}\right)+ \widetilde{D}_{\alpha}(\rho_{MB_{0}^{\prime}}\Vert\tau_{MB_{0}^{\prime}})\\
& = n\frac{\alpha}{\alpha-1}\log\nu_{\alpha}\!\left(
\Theta_{\sigma^{\frac{1-\alpha}{\alpha}}}\circ\mathcal{N}\right)\\
& = n\frac{\alpha}{\alpha-1}\log\sup_{\rho_{A} \in \mathcal{S}(\mathcal{H}_A)}\left\Vert \left(  \Theta
_{\sigma_{B}^{\frac{1-\alpha}{\alpha}}}\circ\mathcal{N}_{A\rightarrow
B}\right)  (\rho_{A})\right\Vert _{\alpha} .
\end{align}
The first inequality follows by combining
 \eqref{normBound} and \eqref{eq:helper1}.
The second inequality follows from monotonicity of the sandwiched R\'enyi relative entropy under the partial trace channel. The third inequality follows by recognizing that $\widetilde{D}_{\alpha}(\rho_{MB_{n-1}B_{n-2}^{\prime}}\Vert\tau_{MB_{n-1}B_{n-2}^{\prime}})$ is the relative entropy at round $n-1$  of $\mathcal{P}_n$, which allows us to apply the above argument inductively. This is a crucial step of the argument and proves a form of additivity similar to that of \eqref{add}. The first equality is a consequence of the fact that $\rho_{MB_0'} = \tau_{MB_0'}$, since no channels have been applied at that point in the protocol. The last equality follows from the main result of \cite{pToqNorms} which allows us to take the supremum over quantum states instead of all operators (furthermore from the fact that quantum states have trace equal to one). Hence, we have that
\begin{align}
\frac{\alpha}{\alpha-1}\log p_{\operatorname{succ}}(\mathcal{P}_{n})+\log L
& \leq n\frac{\alpha}{\alpha-1}\log\sup_{\rho_{A} \in \mathcal{S}(\mathcal{H}_A)}\left\Vert \left(  \Theta
_{\sigma_{B}^{\frac{1-\alpha}{\alpha}}}\circ\mathcal{N}_{A\rightarrow
B}\right)  (\rho_{A})\right\Vert _{\alpha} .
\end{align}
 Hence,
we have proven that the upper bound holds for all states $\sigma_B$, so we can conclude that
\begin{align}
\frac{\alpha}{\alpha-1}\log p_{\operatorname{succ}}(\mathcal{P}_{n})+\log L
& \leq n\inf_{\sigma_{B} \in \mathcal{S}(\mathcal{H}_B)}\sup_{\rho_{A}\in \mathcal{S}(\mathcal{H}_A)}\frac{\alpha}{\alpha-1}\log\left\Vert \left(
\Theta_{\sigma_{B}^{\frac{1-\alpha}{\alpha}}}\circ\mathcal{N}_{A\rightarrow
B}\right)  (\rho_{A})\right\Vert _{\alpha}\\
& = n\widetilde{K}_{\alpha} (\mathcal{N})\\
& =n\widetilde{\chi}_{\alpha}(\mathcal{N}).
\end{align}
The first equality comes from the definition of the $\alpha$-information radius, and the second equality is due to Lemma \ref{chiK}.

Now, the protocol uses the channel $n$ times and the rate $R$ of $\mathcal{P}_n$ is defined to be the number of bits per channel use, so that $R \leq \frac{\log L}{n}$. This allows us to introduce $R$ into the inequality:
\begin{equation}
\frac{1}{n}\log p_{\operatorname{succ}}(\mathcal{P}_{n})\leq-\frac{\alpha-1}{\alpha}(R-\widetilde{\chi}_{\alpha}(\mathcal{N})).
\end{equation}
Since this is true for all $\alpha>1$, we can take a supremum over $\alpha>1$ and arrive at the bound stated
in \eqref{eq:main-bound-1}.
\let\qed\relax
\end{proof}

The strong converse itself now follows from
Theorem~\ref{tehThm} and Lemma \ref{chiLims}.

\begin{corollary}
  [Strong Converse]
The probability of success of any sequence of protocols which use an entanglement-breaking channel with classical feedback at a rate greater than the classical capacity is bounded from above by a decaying exponential. The same is true for arbitrary channels with separable inputs.
\end{corollary}
\begin{proof}
 Recall that since $\widetilde{D}_\alpha$ is monotonically increasing in $\alpha$, so is $\widetilde{\chi}_\alpha$. This along with Lemma~\ref{chiLims} implies
\begin{align}
  \inf_{\alpha>1} \widetilde{\chi}_\alpha(\mathcal{N}) & = \lim_{\alpha\searrow 1} \widetilde{\chi}_\alpha(\mathcal{N})
   = \chi(\mathcal{N}) .
  \label{}
\end{align}
Hence, if $R>\chi(\mathcal{N})$, by the continuity of
 $\widetilde{\chi}_\alpha(\mathcal{N})$ as $\alpha\searrow 1$,
 there exists a value of
$\alpha > 1$ such that $\frac{\alpha-1}{\alpha} \left( R- \widetilde{\chi}_\alpha(\mathcal{N}) \right) >0 $. Then the bound in \eqref{eq:main-bound-1} implies an exponential decay of the success probability.
\end{proof}

\section{Conclusion}
By studying the classical capacity of a quantum channel, quantum information theorists found that an implication of quantum mechanics for classical communication is the possibility of superior encoding schemes that use entanglement, a uniquely quantum phenomenon \cite{notAdd}. However, entanglement does not give an advantage for every channel. The results in this and previous papers show that for entanglement-breaking channels, no communication protocols, even with classical feedback, can take advantage of this extra resource. They cannot use it to increase their capacity, nor, as proved in this paper, to even lift the exponentially decaying ceiling on their success probabilities. We perhaps expect this since by definition EB channels destroy entanglement, and classical feedback channels cannot create entanglement. This is the guiding principle behind our proof, in particular the key fact that the transmitter and receiver states are separable throughout the protocol. 

More generally, it is known that classical feedback can give a large boost to the classical capacity of channels which are not entanglement-breaking in at least two different ways: first,
there exist channels for which the Holevo information is small but the classical capacity with feedback included can be quite large
\cite{PhysRevLett.96.150502}. Similarly, there exist channels for which the classical capacity is small but becomes large when classical feedback is available
\cite{PhysRevLett.103.120503}.
Thus, showing that the Holevo information is a strong-converse bound suggests that for feedback not to help, we require special channels, such as the entanglement-breaking ones.

Going in the opposite direction, it is known that EB channels form a proper superset of classical channels \cite{KHH2012}. We thus obtain the result of \cite{PV} as a direct corollary. Furthermore, since the original posting of our paper as arXiv:1506.02228, an open question that we posed has now been answered --- the strong converse exponent from Theorem~\ref{tehThm} is tight, due to the results in \cite{MO14}. Hence, we obtain as a special case from this and our result the classical results of \cite{A78,1056003,CK82} as well.

A possible direction for future research is to ask the same question for Hadamard channels, which are defined to be complements of EB channels. To define the complementary channel, we first explain the interpretation of channel as a model for open quantum dynamics. That is, a channel from system $A$ to system $B$ can be interpreted as the restriction of a unitary interaction with larger system $BR$, where $R$ is referred to as the ``environment.'' The complementary channel is the map from $A$ to $BR$ tracing out the $B$ system instead and is itself a channel from $A$ to $R$. Hence, Hadamard channels break any entanglement with the environment system that is traced out and is thus related to our guiding principle. The strong converse has already been proved for Hadamard channels \cite{WWY13}, but with the addition of classical feedback, even a weak converse has yet to be proved. Examples of Hadamard channels include generalized dephasing channels, cloning channels, and the Unruh channel \cite{Hdm}.

Finally, we remark here that it should be possible to use the methods given here and in \cite{QFBStrong} in order to characterize a particular adaptive hypothesis testing scenario. Suppose that the goal is to distinguish an entanglement-breaking channel from a replacement channel
by means of adaptive, separability-preserving channels. Then the optimal strong converse exponent should be given in terms of a quantity similar to that in \eqref{eq:main-bound-1}. However, we leave the details for future work.

\bigskip

\textbf{Acknowledgements.} We are grateful to Patrick Hayden, Milan Mosonyi, and Graeme Smith for insightful
discussions about the topic of this paper. We also thank an anonymous referee for many helpful suggestions for improving the paper. MMW acknowledges support from startup funds from the
Department of Physics and Astronomy at LSU, the NSF\ under Award
No.~CCF-1350397, and the DARPA Quiness Program through US Army Research Office
award W31P4Q-12-1-0019. DD acknowledges support from a Stanford Graduate Fellowship.

\appendix

\section{Additivity}

\label{app:additivity}

Here we justify the statement in \eqref{add}.
The backward implication is trivial. As for the forward implication, we first note that given $n$, for all $k \in \mathbb{Z}^+$,
\begin{equation}
  \frac{1}{n} \chi(\mathcal{N}^{\otimes n}) \leq \frac{1}{kn} \chi(\mathcal{N}^{\otimes kn}) ,
  \label{}
\end{equation}
because we can split the $kn$ channels into $k$ blocks of $n$ channels and encode each block independently. In particular, $\frac{1}{n} \chi(\mathcal{N}^{\otimes n}) \geq \chi(\mathcal{N})$. To show the opposite inequality, we assume the contrapositive: $\frac{1}{n} \chi(\mathcal{N}^{\otimes n}) > \chi(\mathcal{N})$ for some $n$. However, this means  
\begin{equation}
  \frac{1}{kn} \chi(\mathcal{N}^{\otimes kn}) - \chi(\mathcal{N}) \geq \frac{1}{n}\chi(\mathcal{N}^{\otimes n}) - \chi(\mathcal{N}) >0  ,
  \label{}
\end{equation}
for all $k\in \mathbb{Z}^+$ and thus the limit in \eqref{capLim} does not converge to $\chi(\mathcal{N})$.

\section{Weak converse and finite bounds}

\label{app:bowen}

We recall for motivation and the reader's convenience the argument for the weak converse from \cite{EBFB}. 
\begin{theorem}
Let $\mathcal{P}_n$ be a protocol for classical feedback-assisted classical communication over an entanglement-breaking channel $\mathcal{N}$ such that it uses the channel $n$ times, has communication rate $R$, and has average probability of decoding error $\varepsilon$. Then, it satisfies the following inequality:
 \begin{equation}
  R\leq\chi(  \mathcal{N})  +g(  n,\varepsilon),
 \end{equation}
where  $\chi(\mathcal{N})$ is the Holevo information of the channel and $g(n,\varepsilon)$ is a real valued function such that
 \begin{equation}
   \lim_{\varepsilon\searrow0} \lim_{n\rightarrow\infty}g(  n,\varepsilon)  =0 .
 \end{equation}
The same is true for a protocol for communication over an arbitrary channel given that the encoder does not entangle inputs across different uses of the channel.
\end{theorem}
\begin{proof}
We use the notation in Section \ref{subsec:protocol} and take a general approach that is, for instance, presented in Section 21.5.1 of \cite{wilde2011classical}. Let $\overline{\Phi}$ denote the following ``shared randomness'' state:%
\begin{equation}
\overline{\Phi}_{M\hat{M}}=\frac{1}{\left\vert \mathcal{M}\right\vert }%
\sum_{m}\left\vert m\right\rangle \left\langle m\right\vert _{M}%
\otimes\left\vert m\right\rangle \left\langle m\right\vert _{\hat{M}},
\end{equation}
where $\left\vert \mathcal{M}\right\vert $ is the size of the message set
$\mathcal{M}$. Now, suppose the information processing task is common randomness generation instead of classical communication. If we require the state $\rho$ at the end of the protocol to be $\varepsilon'$-close to $\overline{\Phi}$
\begin{equation}
  \lVert \overline{\Phi}  - \rho \rVert_1 \le \varepsilon',
  \label{}
\end{equation}
the Fannes-Audenaert inequality for mutual information \cite{fanIneq, audIneq} gives
\begin{align}
nR &  =I(M;\hat{M})_{\overline{\Phi}}\\
&  \leq I(M;\hat{M})_{\rho}+f(  n,\varepsilon') ,
\end{align}
where $f(  n,\varepsilon')  $ is some continuous function of $n$ and
$\varepsilon'$ with the property: $\lim_{\varepsilon'\searrow0}%
\lim_{n\rightarrow\infty}\frac{1}{n}f(  n,\varepsilon')  =0$.
Continuing, we have%
\begin{align}
I(M;\hat{M})_{\rho} &  \leq I(  M;B_{n}B_{n-1}^{\prime})
_{\rho}\label{eq:1st-steps}\\
&  =I(  M;B_{n-1}^{\prime})  _{\rho}+I(  M;B_{n}%
|B_{n-1}^{\prime})  _{\rho},
\end{align}
where we applied the data processing inequality and the chain rule for
conditional quantum mutual information. We use the chain rule again to get
\begin{equation}
I(  M;B_{n}|B_{n-1}^{\prime})  _{\rho}\leq I(
MB_{n-1}^{\prime};B_{n})  _{\rho}.
\end{equation}
From Section \ref{subsec:protocol}, if $\mathcal{N}$ is entanglement-breaking or if Alice uses separable inputs, Alice and Bob's systems are always separable. Hence, the global state before the $n$th channel use can be written as%
\begin{align}
  \rho_{MA_{n}^{\prime}A_{n}B_{n-1}^{\prime}} & = \sum_m p_M(m) \vert m \rangle \langle m\vert_M \otimes \rho_{A_{n}' A_n B_{n-1}'}^{m} \\
& = \sum_{m}p_{M}\left(
m\right)  \left\vert m\right\rangle \left\langle m\right\vert _{M}\otimes
\sum_{y}p_{Y|M}(  y|m)  \rho_{A_{n}^{\prime} A_n}^{m,y} 
\otimes\rho_{B_{n-1}^{\prime}}^{m,y}.
\end{align}
Hence, the state after the channel is given by
\begin{equation}
  \rho_{MA_{n}' B_n B_{n-1}'} = 
\sum_{m}p_{M}\left(
m\right)  \left\vert m\right\rangle \left\langle m\right\vert _{M}\otimes
\sum_{y}p_{Y|M}(  y|m)  \mathcal{N}_{A_n \rightarrow B_n} \!\left(\rho_{A_{n}^{\prime} A_n}^{m,y} \right)
\otimes\rho_{B_{n-1}^{\prime}}^{m,y}.
  \label{}
\end{equation}
We introduce an auxiliary system $Y$ that labels the separable sum over the index $y$. That is, $Y$ is chosen such that the global state is
\begin{equation}
\rho_{MYA_{n}^{\prime}B_{n}B_{n-1}^{\prime}}\equiv\sum_{m,y}p_{M}\left(
m\right)  p_{Y|M}(  y|m)  \left\vert m\right\rangle \left\langle
m\right\vert _{M}\otimes\left\vert y\right\rangle \left\langle y\right\vert
_{Y}\otimes \mathcal{N}_{A_n \rightarrow B_n} \!\left(\rho_{A_{n}^{\prime} A_n}^{m,y} \right)
  \otimes\rho_{B_{n-1}^{\prime}}^{m,y}.
\end{equation}
We trace over $A_{n}'$ to get
\begin{equation}
\rho_{MYB_{n}B_{n-1}^{\prime}} = \sum_{m,y}p_{M}\left(
m\right)  p_{Y|M}(  y|m)  \left\vert m\right\rangle \left\langle
m\right\vert _{M}\otimes\left\vert y\right\rangle \left\langle y\right\vert
_{Y}\otimes \rho_{B_n}^{m,y} 
  \otimes\rho_{B_{n-1}^{\prime}}^{m,y} ,
  \label{}
\end{equation}
where
\begin{equation}
  \rho_{B_n}^{m,y} \equiv \Tr_{A_{ n }'}  \left(\mathcal{N}_{A_n \rightarrow B_n} \left(\rho_{A_{n}^{\prime} A_n}^{m,y} \right)\right) = \mathcal{N}_{A_n \rightarrow B_n} \left(\rho_{A_n}^{m,y}\right) .
  \label{output}
\end{equation}
This allows us to argue
\begin{align}
I(  MB_{n-1}^{\prime};B_{n})  _{\rho} &  \leq I(
MYB_{n-1}^{\prime};B_{n})  _{\rho}\\
&  =I(  MY;B_{n})  _{\rho}+I(  B_{n-1}^{\prime}%
;B_{n}|MY)  _{\rho}\\
&  =I(  MY;B_{n})  _{\rho}\\
& \leq \chi(\mathcal{N}) ,
\end{align}
where the first inequality is from data processing, the first equality from
the chain rule, and the second equality because the $B_{n-1}' B_n$ system
is in a product state when conditioning on $M$ and $Y$. Finally, given equation \eqref{output}, the state $\rho_{MYB_n}$ is a classical-quantum state of the form:
\begin{align}
   \rho_{MYB_{n}} & = \Tr_{ B_{n-1}'} \left(\rho_{MY  B_n B_{n-1}'} \right)\\
   & = \sum_{m,y} p_M \left(
m\right)  p_{Y|M}(  y|m)  \left\vert m\right\rangle \left\langle
m\right\vert _{M}\otimes\left\vert y\right\rangle \left\langle y\right\vert
_{Y}\otimes \mathcal{N}_{A_n \rightarrow B_n} \left(\rho_{A_n}^{m,y}\right) .
  \label{}
\end{align}
Hence, the definition of the Holevo information of the channel $\mathcal{N}$ gives us final inequality.
Putting things together, we find that%
\begin{align}
I(M;\hat{M})_{\rho} &  \leq\chi(  \mathcal{N})  +I(
M;B_{n-1}^{\prime})  _{\rho}\\
&  \leq\chi(  \mathcal{N})  +I(  M;B_{n-1}B_{n-2}^{\prime
})  _{\rho},
\end{align}
where the last inequality follows from the data processing inequality. But now, we recognize the quantity $I(  M;B_{n-1}B_{n-2}^{\prime})
_{\rho}$ is of the same form as $I(  M;B_{n}B_{n-1}^{\prime})
_{\rho}$ in \eqref{eq:1st-steps}, so that we can iterate through the same
sequence of arguments to get%
\begin{equation}
I(M;\hat{M})_{\rho}\leq2\chi(  \mathcal{N})  +I(
M;B_{n-2}B_{n-3}^{\prime})  _{\rho}.
\end{equation}
Continuing all the way back to the first channel use, we find%
\begin{equation}
I(M;\hat{M})_{\rho}\leq n\chi(  \mathcal{N}) 
\end{equation}
since $I(M;B_0')=0$ (see \eqref{initial}). Thus, we conclude
\begin{equation}
R\leq\chi(  \mathcal{N})  +\frac{1}{n}f(  n,\varepsilon')
.
\end{equation}
Now, the rate for classical communication with $n$ channel uses and error $\varepsilon$ is at most that of randomness generation with $n$ uses and error $\varepsilon'$, where $\varepsilon'$ is some function of $\varepsilon$ such that $\lim_{\varepsilon \searrow 0} \varepsilon' = 0$. Hence, $g(n,\varepsilon) \equiv \frac{1}{n}f(n,\varepsilon')$ is what we need.
\end{proof}

\begin{corollary}
  [Weak Converse]
  The classical capacity of entanglement-breaking channels with classical feedback is given by the Holevo information. The same is true for arbitrary channels with separable inputs.
\end{corollary}
\begin{proof}
  This is immediate.
\end{proof}

\section{Limit value of the $\alpha$-Holevo information}
\label{sec:pfLemma}
The arguments here are essentially the same as those in
\cite{MH,QFBStrong}, along with an additional insight
from \cite[Appendix A]{TWW14}.
We first recall a minimax result due to \cite[Corollary A2]{MH}.
\begin{lemma}
  Suppose a function $f:X \times Y \rightarrow \widebar{\mathbb{R}}$, where $X$ is a compact topological space, $Y$ is a subset of $\mathbb{R}$, and $\widebar{\mathbb{R}}$ is the extended real numbers, satisfies
  \begin{enumerate}
    \item $\forall y\in Y$, $f(\cdot,y)$ is lower semicontinuous.
    \item $\forall x\in X$, $f(x,\cdot)$ is monotonic.
  \end{enumerate}
  Then,
  \begin{equation}
    \inf_{x\in X} \sup_{y\in Y} f(x,y) = \sup_{y\in Y} \inf_{x\in X}  f(x,y)
    \label{}.
  \end{equation}
  If the first condition was instead
  \begin{enumerate}
    \item $\forall y\in Y$, $f(\cdot,y)$ is upper semicontinuous,
  \end{enumerate}
  then
  \begin{equation}
    \inf_{y\in Y} \sup_{x\in X}  f(x,y)= \sup_{x\in X} \inf_{y\in Y} f(x,y)
    \label{}.
  \end{equation}
  
  \label{topStuff}
\end{lemma}
We now prove Lemma \ref{chiLims}.
\begin{proof}
  For fixed $\rho_{XR'}$ and $\alpha$, $\sigma_{R'}
  \mapsto \widetilde{D}_\alpha(\rho_{XR'},\rho_X \otimes \sigma_{R'})$ is  lower semicontinuous (see, e.g., \cite[Appendix A]{QFBStrong}). Furthermore, for fixed $\rho_{XR'}$ and
  $\sigma_{R'}$, $\alpha \mapsto \widetilde{D}_\alpha(\rho_{XR'},\rho_X \otimes \sigma_{R'})$ is monotone non-decreasing.  We can therefore invoke Lemma \ref{topStuff} with $X= \mathcal{S}(\mathcal{H}_{R'})$ and $Y = (0,1)$. We use this, the definition of the Holevo information, and properties of the quantum relative entropy to find the
  following one-sided limit of the $\alpha$-Holevo information of a quantum channel $\mathcal{N}:\mathcal{B}(\mathcal{H}_R) \rightarrow \mathcal{B}(\mathcal{H}_{R'})$:
  \begin{align}
    \lim_{\alpha\nearrow 1} \widetilde{\chi}_\alpha(\mathcal{N}) & = \sup_{\alpha \in (0,1)} \widetilde{\chi}_\alpha(\mathcal{N})\\
    & = \sup_{\alpha \in(0,1)}\sup_{\left\{p_X(x),\rho_x \right\}} \widetilde{\chi}_\alpha(\left\{ p_X(x),\mathcal{N}(\rho_x) \right\}) \\
    & = \sup_{\left\{p_X(x),\rho_x \right\}}\sup_{\alpha \in(0,1)} \inf_{\sigma_{R'} \in \mathcal{S}(\mathcal{H}_{R'})} \widetilde{D}_\alpha(\rho_{XR'}\Vert\rho_X\otimes\sigma_{R'})\\
    & = \sup_{\left\{p_X(x),\rho_x \right\}} \inf_{\sigma_{R'} \in \mathcal{S}(\mathcal{H}_{R'}) }\sup_{\alpha \in(0,1)}  \widetilde{D}_\alpha(\rho_{XR'}\Vert\rho_X\otimes\sigma_{R'})\\
    & = \sup_{\left\{p_X(x),\rho_x \right\}} \inf_{\sigma_{R'} \in \mathcal{S}(\mathcal{H}_{R'})} D(\rho_{XR'}\Vert\rho_X\otimes\sigma_{R'})\\
    & = \sup_{\left\{p_X(x),\rho_x \right\}} I(X:R')_{\rho}\\
    & = \chi(\mathcal{N}) ,
    \label{}
  \end{align}
where
\begin{equation}
  \rho_{XR'} \equiv  \sum_x p_X(x) \vert x \rangle \langle x \vert_X \otimes \left[ \mathcal{N}(\rho_x) \right]_{R'} .
  \label{}
\end{equation}
The fourth equality follows from applying Lemma~\ref{topStuff}. The fifth follows from \eqref{eq:sandwiched-to-Umegaki} and the fact that the sandwiched R\'enyi relative entropy is monotone non-decreasing in $\alpha$.

To find the other limit, we note that for fixed $\rho_{R}$ and $\sigma_{R^{\prime}}$, $\alpha\mapsto
\widetilde{D}_{\alpha}(\mathcal{N}(\rho_{R})\Vert\sigma_{R^{\prime}})$ is
monotone non-decreasing. We then use an idea from \cite[Appendix A]{TWW14}. Note that for fixed $\alpha$
and positive definite $\sigma_{R^{\prime}}$, $\rho_{R}\mapsto\widetilde
{D}_{\alpha}(\mathcal{N}(\rho_{R})\Vert\sigma_{R^{\prime}})$ is continuous (and thus upper semicontinuous).
Let $\varepsilon \in (0,1)$ and define $\sigma\left(
\varepsilon\right)  \equiv\left(  1-\varepsilon\right)  \sigma+\varepsilon\pi
$, where $\pi$ is the maximally mixed state. Thus
$\sigma(\varepsilon)$ is positive definite for $\sigma$ positive semidefinite. Consider that%
\begin{equation}
\sigma(  \varepsilon)  \geq\left(  1-\varepsilon\right)  \sigma,
\end{equation}
which implies that%
\begin{equation}
D(  \rho\Vert\sigma(  \varepsilon)  )  \leq D(
\rho\Vert\left(  1-\varepsilon\right)  \sigma)  =D(  \rho
\Vert\sigma)  -\log\left(  1-\varepsilon\right)  .
\label{eq:dominating-prop}
\end{equation}
(This is because of the well known fact that $\sigma \leq \sigma'$ implies that $D(\rho\Vert \sigma) \geq 
D(\rho\Vert \sigma')$.)
Hence, using Lemma \ref{topStuff}, Lemma \ref{chiK}, and the identification of
the Holevo information as an information radius \cite{OPW, SBW}, we can
conclude the following chain of equalities:%
\begin{align}
\lim_{\alpha\searrow1}\widetilde{\chi}_{\alpha}(\mathcal{N}) &  =\inf
_{\alpha>1}\widetilde{\chi}_{\alpha}(\mathcal{N})\\
&  =\inf_{\alpha>1}\widetilde{K}_{\alpha}(\mathcal{N})\\
&  =\inf_{\alpha>1}\inf_{\sigma_{R^{\prime}}\in\mathcal{S}(\mathcal{H}%
_{R^{\prime}})}\sup_{\rho_{R}\in\mathcal{S}(\mathcal{H}_{R})}\widetilde
{D}_{\alpha}\left(  \mathcal{N}(\rho_{R})\Vert\sigma_{R^{\prime}}\right)  \\
&  \leq\inf_{\alpha>1}\inf_{\sigma_{R^{\prime}}\in\mathcal{S}(\mathcal{H}%
_{R^{\prime}})}\sup_{\rho_{R}\in\mathcal{S}(\mathcal{H}_{R})}\widetilde
{D}_{\alpha}\left(  \mathcal{N}(\rho_{R})\Vert\sigma(  \varepsilon
)  _{R^{\prime}}\right)  \\
&  =\inf_{\sigma_{R^{\prime}}\in\mathcal{S}(\mathcal{H}_{R^{\prime}})}%
\inf_{\alpha>1}\sup_{\rho_{R}\in\mathcal{S}(\mathcal{H}_{R})}\widetilde
{D}_{\alpha}\left(  \mathcal{N}(\rho_{R})\Vert\sigma(  \varepsilon
)  _{R^{\prime}}\right)  \\
&  =\inf_{\sigma_{R^{\prime}}\in\mathcal{S}(\mathcal{H}_{R^{\prime}})}%
\sup_{\rho_{R}\in\mathcal{S}(\mathcal{H}_{R})}\inf_{\alpha>1}\widetilde
{D}_{\alpha}\left(  \mathcal{N}(\rho_{R})\Vert\sigma(  \varepsilon
)  _{R^{\prime}}\right)  \\
&  =\inf_{\sigma_{R^{\prime}}\in\mathcal{S}(\mathcal{H}_{R^{\prime}})}%
\sup_{\rho_{R}\in\mathcal{S}(\mathcal{H}_{R})}D(  \mathcal{N}(\rho
_{R})\Vert\sigma(  \varepsilon)  _{R^{\prime}})  \\
&  \leq\inf_{\sigma_{R^{\prime}}\in\mathcal{S}(\mathcal{H}_{R^{\prime}})}%
\sup_{\rho_{R}\in\mathcal{S}(\mathcal{H}_{R})}D(  \mathcal{N}(\rho
_{R})\Vert\sigma_{R^{\prime}})  -\log(  1-\varepsilon)  \\
&  =\chi(\mathcal{N})-\log\left(  1-\varepsilon\right)  .
\end{align}
The first inequality follows because there is an infimum with respect to $\sigma_{R'}$, such that replacing 
$\sigma_{R'}$ with $\sigma(\varepsilon)_{R'}$ can never decrease the quantity. The fifth equality follows from applying Lemma~\ref{topStuff}. The sixth equality follows from \eqref{eq:sandwiched-to-Umegaki} and the fact that the sandwiched R\'enyi relative entropy is monotone non-decreasing in $\alpha$. The last inequality follows from \eqref{eq:dominating-prop}. The last equality uses $\chi(\mathcal{N}) = K(\mathcal{N})$.

Given that $\varepsilon>0$ was arbitrary, we can conclude that $\lim
_{\alpha\searrow1}\widetilde{\chi}_{\alpha}(\mathcal{N})\leq\chi(\mathcal{N}%
)$. Combined with the fact that $\widetilde{\chi}_{\alpha}(\mathcal{N}%
)\geq\chi(\mathcal{N})$ for all $\alpha\geq1$, we conclude that%
\begin{equation}
\lim_{\alpha\searrow1}\widetilde{\chi}_{\alpha}(\mathcal{N})=\chi
(\mathcal{N}).
\end{equation}
This also implies that
\begin{equation}
\lim_{\alpha\searrow1}\widetilde{K}_{\alpha}(\mathcal{N})=K
(\mathcal{N}).
\end{equation}

We now show that
$\lim_{\alpha\nearrow1}\widetilde{K}_{\alpha}(\mathcal{N})=K
(\mathcal{N})
$,
which follows from a similar line of reasoning:
\begin{align}
  \lim_{\alpha \nearrow 1} \wtil{K}_\alpha(\mathcal{N}) & = \sup_{\alpha \in (0,1)} \wtil{K}_\alpha(\mathcal{N}) \\
  & = \sup_{\alpha \in (0,1)} \inf_{\sigma_{R'} \in \mathcal{S}(\mathcal{H}_{R'})} \sup_{\rho_R \in \mathcal{S}(\mathcal{H}_R)} \wtil{D}_\alpha(\mathcal{N}(\rho_R) \Vert \sigma_{R'})\\
  & =  \inf_{\sigma_{R'} \in \mathcal{S}(\mathcal{H}_{R'})} \sup_{\alpha \in (0,1)}\sup_{\rho_R \in \mathcal{S}(\mathcal{H}_R)} \wtil{D}_\alpha(\mathcal{N}(\rho_R) \Vert \sigma_{R'})\\
  & =  \inf_{\sigma_{R'} \in \mathcal{S}(\mathcal{H}_{R'})} \sup_{\rho_R \in \mathcal{S}(\mathcal{H}_R)} \sup_{\alpha \in (0,1)}\wtil{D}_\alpha(\mathcal{N}(\rho_R) \Vert \sigma_{R'})\\
  & =  \inf_{\sigma_{R'} \in \mathcal{S}(\mathcal{H}_{R'})} \sup_{\rho_R \in \mathcal{S}(\mathcal{H}_R)} {D}(\mathcal{N}(\rho_R) \Vert \sigma_{R'})\\
  & = K(\mathcal{N}),
  \label{}
\end{align}
where the third equality follows from Lemma~\ref{topStuff} and the fact that for fixed $\alpha$ and $\rho_R$, $\sigma_{R'} \mapsto \wtil{D}_\alpha(\mathcal{N}(\rho_R) \Vert \sigma_{R'})$ is lower semicontinuous and that the pointwise supremum of lower semicontinuous functions is lower semicontinuous. 
\end{proof}

\bibliographystyle{alpha}

\bibliography{Ref}

\end{document}